\documentclass{llncs}
  
\usepackage{amsmath}
\usepackage{amssymb}
\usepackage{times}
\usepackage{tikz} 
\usepackage{booktabs}
\usepackage{url}
\usepackage{changebar}
\usepackage{xspace}
\usepackage{threeparttable}

\usepackage{multirow} \usepackage{rotating}

\newcommand{\hy}{\hbox{-}\nobreak\hskip0pt}

\newcommand{\CCC}{\mathcal{C}}
\newcommand{\EEE}{\mathcal{E}}
\newcommand{\SSS}{\mathcal{S}}

\newcommand{\SB}{\{\,}
\newcommand{\SM}{\;{:}\;}
\newcommand{\SE}{\,\}}
\newcommand{\Card}[1]{|#1|}

\newcommand{\nil}{{\{\}}}

\newcommand{\NP}{\text{\normalfont  NP}}
\newcommand{\coNP}{\text{\normalfont co-NP}}

\newcommand{\W}[1][xxxx]{\text{\normalfont W[#1]}}
\let\phi=\varphi 

\renewenvironment{proof}{\begin{pf}}{\qed\end{pf}}

\newcommand{\class}[1]{\text{\text{\normalfont\sc  #1}}}

\newcommand{\SCHAEFER}{\text{\normalfont\sffamily{Schaefer}}}
\newcommand{\SUBSOLVER}{\mathnormal{\text{\normalfont\sffamily{Subsolver}}}}

\newcommand{\SAT}{\class{SAT}}
\newcommand{\THREESAT}{\class{3SAT}}

\newcommand{\HORN}{\class{Horn}}
\newcommand{\ZEROV}{\class{$0$-Val}}
\newcommand{\ONEV}{\class{$1$-Val}}
\newcommand{\RHORN}{\class{RHorn}}
\newcommand{\TWOCNF}{\class{2CNF}}
\newcommand{\UP}{\class{UP}}
\newcommand{\PL}{\class{PL}}
\newcommand{\CLU}{\class{Clu}}

\newcommand{\FOREST}{\class{Forest}}

\newcommand{\wb}{\mbox{\normalfont wb}}
\renewcommand{\sb}{\mbox{\normalfont sb}}
\newcommand{\db}{\mbox{\normalfont db}}

\newcommand{\set}[1]{\left\{ #1 \right\}}
\newcommand{\myiff}{if and only if\xspace}

\newcommand{\mywrt}{with respect to\xspace}
\newcommand{\ie}{i.e.\xspace}
\newcommand{\etal}{\emph{et al.}\xspace}

\newcommand{\var}{\mbox{var}}

\newcommand{\ol}[1]{\overline{#1}}

\makeatletter
\def\@maketitle{\newpage
 \markboth{}{}%
 \def\lastand{\ifnum\value{@inst}=2\relax
                 \unskip{} \andname\
              \else
                 \unskip \lastandname\
              \fi}%
 \def\and{\stepcounter{@auth}\relax
          \ifnum\value{@auth}=\value{@inst}%
             \lastand
          \else
             \unskip,
          \fi}%
 \begin{center}%
   {\Large \bfseries\boldmath
     \pretolerance=10000
     \@title \par}\vskip .8cm
   \if!\@subtitle!\else {\large \bfseries\boldmath
     \vskip -.65cm
     \pretolerance=10000
     \@subtitle \par}\vskip .8cm\fi
   \setbox0=\vbox{\setcounter{@auth}{1}\def\and{\stepcounter{@auth}}%
     \def\thanks##1{}\@author}%
   \global\value{@inst}=\value{@auth}%
   \global\value{auco}=\value{@auth}%
   \setcounter{@auth}{1}%
   {\lineskip .5em
     \noindent\ignorespaces
     \@author\vskip.35cm}
   {\small\institutename}
   {\vskip.70cm
     \noindent\ignorespaces
     {\em Dedicated to Mike Fellows on the occasion of his 60th birthday.}
     \vskip.35cm}

 \end{center}%
}
\makeatother

\begin{document}

\title{Backdoors to Satisfaction%
\thanks{Research supported by the European Research Council (ERC), project COMPLEX REASON 239962.}}

\author{Serge Gaspers \and Stefan Szeider}

\institute{
  Institute of Information Systems,
  Vienna University of Technology,
  A-1040 Vienna, Austria
  \email{gaspers@kr.tuwien.ac.at, stefan@szeider.net}
}

\maketitle

\begin{abstract}
  A backdoor set is a set of variables of a propositional
  formula such that fixing the truth values of the variables in the
  backdoor set moves the formula into some polynomial-time decidable
  class. If we know a small backdoor set we can reduce the question of
  whether the given formula is satisfiable to the same question for
  one or several easy formulas that belong to the tractable class
  under consideration. 
  In this survey we review parameterized complexity results for
  problems that arise in the context of backdoor sets, such as the
  problem of finding a backdoor set of size at most $k$, parameterized
  by $k$. We also discuss recent results on backdoor sets for problems
  that are beyond~NP.
\end{abstract}

\section{Introduction}

Satisfiability (SAT) is the classical problem of
determining whether a propositional formula in conjunctive normal form
(CNF) has a satisfying truth assignment.
The famous Cook-Levin Theorem \cite{Cook71,Levin73}, stating that SAT is
NP-complete, placed satisfiability as the cornerstone of complexity
theory.  Despite its seemingly specialised nature, satisfiability has
proved to be extremely useful in a wide range of different disciplines,
both from the practical as well as from the theoretical point of view.
Satisfiability provides a powerful and general formalism for solving
various important problems including hardware and software verification
and planning
\cite{BjesseLeonardMokkedem01,PrasadBiereGupta05,VelevBryant03,KautzSelman92}.
Satisfiability is
the core of many reasoning problems in automated deduction; for
instance, the package dependency management for the OpenSuSE Linux
distribution and the autonomous controller for NASA's Deep Space One
spacecraft are both based on satisfiability~\cite{LeberreParrain08,Weld99}.
Over the last two decades, SAT-solvers have become amazingly
successful in solving formulas with hundreds of thousands of variables that encode
problems arising from various application areas, see, e.g.,
\cite{GomesKautzSabharwalSelman08}.  Theoretical performance
guarantees, however, are far from explaining this empirically observed
efficiency. In fact, there is an enormous gap between theory and
practice.  To illustrate it with numbers, take the exponential factor
$1.308^n$ of the currently fastest known exact 3SAT algorithm~\cite{Hertli11}.
Already for $n=250$ variables this number exceeds by far the expected
lifetime of the sun in nanoseconds.

%

\paragraph{Hidden Structure and Parameterized Complexity} The discrepancy between theory and
practice can be explained by the presence of a certain ``hidden
structure'' in real-world problem instances.  It is a widely accepted
view that the structure of real-world problem instances makes the
problems easy for heuristic solvers. However, classic worst-case
analysis is not particularly well-suited to take this hidden structure
into account.  The classical model is one-dimensional, where only one
aspect of the input (its size in bits, or the number of variables for
a SAT formula) is taken into account, and it does not differentiate whether or not the instance is
otherwise well-structured.

Parameterized Complexity, introduced by Mike Fellows together with Rod
Downey offers a two-dimensional theoretical setting. The first dimension
is the input size as usual, the second dimension (the parameter) allows
to take structural properties of the problem instance into account. The
result is a more fine-grained complexity analysis that has the potential
of being more relevant to real-world computation while still admitting a
rigorous theoretical treatment and firm algorithmic performance
guarantees.

There are various ways of defining the ``hidden structure'' in a problem
instance, yielding various ways to parameterize a problem.

\paragraph{Islands of Tractability} One way of coping with the high
complexity of important problems within the framework of classical
complexity is the identification of tractable sub-problems, \ie, of
classes of instances for which the problem can be solved in
polynomial time. Each class represents an ``island of tractability''
within an ocean of intractable problems.  For the
satisfiability problem, researchers have identified dozens of such
islands -- one could speak of an archipelago of tractability.
 
Usually it is quite unlikely that a real-world instance belongs to a
known island of tractability, but it may be close to one. A very natural
and humble way of parameterizing a problem is hence to take the distance
to an island of tractability as a parameter.  Guo \etal
\cite{GuoHuffnerNiedermeier04} called this approach ``distance to
triviality''. For SAT, the distance is most naturally measured in terms
of the smallest number variables that need to be instantiated or deleted
such that the instance gets moved to an island of tractability. Such
sets of variables are called \emph{backdoor sets} because once we know a
small backdoor set we can solve the instance efficiently. Thus backdoor
sets provide a ``clever reasoning shortcut'' through the search space
and can be used as an indicator for the presence of a hidden structure
in a problem instance.  Backdoor sets where independently introduced by
Crama \etal \cite{CramaEkinHammer97} and by Williams
\etal~\cite{WilliamsGomesSelman03}, the latter authors coined the term
``backdoor''.

The backdoor set approach to a problem consists of two steps: first a
small backdoor set is computed (\emph{backdoor detection}), second the
backdoor set is used to solve the problem at hand (\emph{backdoor
  evaluation}).  It is hence natural to consider an upper bound on the
size of a smallest backdoor set as a parameter for both backdoor
detection and backdoor evaluation.



\section{Satisfiability}

The \emph{propositional satisfiability problem} (SAT) was the first
problem shown to be NP-hard \cite{Cook71,Levin73}. Despite its hardness, SAT
solvers are increasingly leaving their mark as a general-purpose tool in
areas as diverse as software and hardware verification, automatic test
pattern generation, planning, scheduling, and even challenging problems
from algebra~\cite{GomesKautzSabharwalSelman08}.

A \emph{literal} is a propositional variable $x$ or a negated variable
$\neg x$. We also use the notation $x=x^1$ and $\neg x= x^0$. A
\emph{clause} is a finite set literals that does not contain a
complementary pair $x$ and $\neg x$. A propositional formula in
conjunctive normal form, or \emph{CNF formula} for short, is a set of
clauses. An $r$CNF formula is a CNF formula where each clause
contains at most $r$ literals. For a clause $C$
we write $\var(C)=\SB x \SM x\in C$ or $\neg x \in C\SE$, and for a CNF
formula $F$ we write $\var(F)=\bigcup_{C\in F} \var(C)$.  An $r$-CNF
formula is a CNF formula where each clause contains at most $r$
literals. For a set $S$ of literals we write $\ol{S}=\SB x^{1-\epsilon}
\SM x^\epsilon \in S\SE$. We call a clause $C$ \emph{positive} if
$C=\var(C)$ and \emph{negative} if $\ol{C}=\var(C)$.

For a set $X$ of propositional variables we denote by $2^X$ the set of
all mappings $\tau:X\rightarrow \set{0,1}$, the truth assignments on $X$.
For $\tau\in 2^X$ we let 
$\text{true}(\tau)=\SB x^{\tau(x)}\SM x\in X\SE$ and
$\text{false}(\tau)=\SB x^{1-\tau(x)}\SM x\in X \SE$ be the sets of literals set
by $\tau$ to $1$ and $0$, respectively. 
Given a CNF formula $F$ and a truth assignment $\tau \in 2^X$ we define
$F[\tau]=\SB C\setminus \text{false}(\tau) \SM C\in F,\ C\cap\text{true}(\tau)=
\emptyset\SE$. If $\tau\in 2^{\set{x}}$ and $\epsilon=\tau(x)$, we simple
write $F[x=\epsilon]$ instead of $F[\tau]$.

A CNF formula $F$ is \emph{satisfiable} if there is some $\tau\in
2^{\var(F)}$ with $F[\tau]=\emptyset$, otherwise $F$ is
\emph{unsatisfiable}. Two CNF formulas are \emph{equisatisfiable} if
either both are satisfiable, or both are unsatisfiable.  SAT is the
$\NP$\hy complete problem of deciding whether a given CNF formula is
satisfiable~\cite{Cook71,Levin73}.

\paragraph{Islands of Tractability and Backdoors}
Backdoors are defined \mywrt a fixed class $\CCC$ of CNF
formulas, the \emph{base class} (or \emph{target class}, or more more
figuratively, \emph{island of tractability}). From a base class we
require the following properties: (i)~$\CCC$ can be recognized in
polynomial time, (ii)~the satisfiability of formulas in $\CCC$ can be
decided in polynomial time, and (iii)~$\CCC$ is closed under isomorphisms
(\ie, if two formulas differ only in the names of their variables,
then either both or none belong to $\CCC$).

Several base classes considered in this survey also satisfy additional
properties. Consider a class $\CCC$ of CNF formulas.  $\CCC$~is
\emph{clause-induced} if it is closed under subsets, \ie, if $F\in
\CCC$ implies $F'\in \CCC$ for each $F'\subseteq F$. $\CCC$~is
\emph{clause-defined} if for each CNF formula $F$ we have $F\in \CCC$
if and only if $\set{C}\in \CCC$ for all clauses $C\in F$.  $\CCC$~is
closed under \emph{variable-disjoint union} if for any two CNF
formulas $F_1,F_2\in \CCC$ with $\var(F_1)\cap \var(F_2)=\emptyset$,
also $F_1\cup F_2\in \CCC$.  $\CCC$~is \emph{self-reducible} if for
any $F \in \CCC$ and any partial truth assignment $\tau$, also
$F[\tau]\in \CCC$.


A \emph{strong} \emph{$\CCC$-backdoor set} of a CNF formula $F$ is a set
$B$ of variables such that $F[\tau]\in \CCC$ for each $\tau \in 2^B$.  A
\emph{weak} \emph{$\CCC$-backdoor set} of $F$ is a set $B$ of variables
such that $F[\tau]$ is satisfiable and $F[\tau]\in \CCC$ holds for some
$\tau \in 2^B$.
A \emph{deletion $\CCC$-backdoor set} of $F$ is a set $B$ of variables such that
$F - B \in \CCC$, where $F - B = \set{C \setminus \set{x^0, x^1 : x\in B} : C \in F}$.

If we know a strong $\CCC$-backdoor set of $F$ of
size $k$, we can reduce the satisfiability of $F$ to the satisfiability
of $2^k$ formulas in $\CCC$. 
Thus SAT becomes fixed-parameter tractable in $k$.  If we know a
weak $\CCC$\hy backdoor set of $F$, then $F$ is clearly satisfiable, and
we can verify it by trying for each $\tau \in 2^X$ whether $F[\tau]$ is
in $\CCC$ and satisfiable.
If $\CCC$ is clause-induced, any deletion $\CCC$-backdoor set of $F$
is a strong $\CCC$-backdoor set of $F$. For several base classes, deletion
backdoor sets are of interest because they are easier to detect than strong backdoor sets.
The challenging problem is to find a strong, weak, or deletion
$\CCC$\hy backdoor set of size at most $k$ if it exists. For
each class $\CCC$ of CNF formulas we consider the following decision
problems.

\begin{quote}
  \textsc{Strong $\CCC$-Backdoor Set Detection}\\ \nopagebreak
  \emph{Instance}: A CNF formula $F$ and an integer $k\geq 0$.\\ \nopagebreak
  \emph{Parameter}: The integer $k$.\\ \nopagebreak
  \emph{Question}: Does $F$ have a strong $\CCC$-backdoor set of size
  at most $k$?
\end{quote}
The problems \textsc{Weak $\CCC$-Backdoor Set Detection} and \textsc{Deletion $\CCC$-Backdoor Set Detection} are defined
similarly.

In fact, for the backdoor approach we actually need the functional
variants of these problems, where if a backdoor set of size at most $k$
exists, such a set is computed. However, for all cases considered in
this survey, where backdoor detection is fixed-parameter tractable, the
respective algorithms also compute a backdoor set.

We also consider these problems for formulas with bounded clause lengths.
All such results are stated for 3CNF formulas, but hold, more generally,
for $r$CNF formulas, where $r\ge 3$ is a fixed integer.

\section{Base Classes}

In this section we define the base classes for the SAT problem that we
will consider in this survey.

\subsection{Schaefer's Base Classes}
In his seminal paper, Schaefer~\cite{Schaefer78} classified the
complexity of generalized satisfiability problems in terms of the
relations that are allowed to appear in constraints. For CNF
satisfiability, this yields the following five base classes\footnote{Affine Boolean formulas considered by Schaefer do
  not correspond naturally to a class of CNF formulas, hence we do not
  consider them here}.
\begin{enumerate}
\item \emph{Horn formulas:} CNF formulas where each clause contains at most one
  positive literal.
\item \emph{Anti-Horn formulas:} CNF formulas where each clause contains
  at most one negative literal.
\item \emph{2CNF formulas:} CNF formulas where each clause contains at most
   two literals.
\item \emph{0-valid formulas:} CNF formulas where each clause contains at least
  one negative literal.
\item \emph{1-valid formulas:} CNF formulas where each clause contains at least
  one positive literal.
\end{enumerate}
We denote the respective classes of CNF formulas by $\HORN$, $\HORN^-$,
$\TWOCNF$, $\ZEROV$, and $\ONEV$, and we write $\SCHAEFER=\set{\HORN,
\HORN^-, \TWOCNF, \ZEROV, \ONEV}$. We note that all these classes
are clause-defined, and by Schaefer's Theorem, these are the only
maximal clause-defined base classes. We also note that $\ZEROV$ and
$\ONEV$ are the only two base classes considered in this survey that
are not self-reducible.

\subsection{Base Classes Based on Subsolvers}

State-of-the-art SAT-solvers are based on variants of the so-called
Davis-Logemann-Loveland (DPLL) procedure
\cite{DavisPutnam60,DavisLogemannLoveland62} (see also~\cite{CookMitchell97}).  The DPLL
procedure searches systematically for a satisfying assignment,
applying first \emph{unit propagation} and \emph{pure literal
  elimination} as often as possible. Then, DPLL branches on the truth
value of a variable, and recurses. The algorithms stops if either
there are no clauses left (the original formula is satisfiable) or all
branches of the search lead to an empty clause (the original formula
is unsatisfiable).  Unit propagation takes as input a CNF formula $F$
that contains a ``unit clause'' $\set{x^\epsilon}$ and outputs
$F[x=\epsilon]$.  Pure literal elimination takes as input a CNF
formula $F$ that has a ``pure literal'' $x^\epsilon$, where $x\in \var(F)$
and $x^{1-\epsilon} \notin \bigcup_{C\in F} C$, and outputs $F[x=\epsilon]$.
In both cases $F$ and $F[x=\epsilon]$ are equisatisfiable. If we omit
the branching, we get an incomplete algorithm which decides
satisfiability for a subclass of CNF formulas. Whenever the algorithm
reaches the branching step, it halts and outputs ``give up''.  This
incomplete algorithm is an example of a ``subsolver'' as considered by
Williams \etal\cite{WilliamsGomesSelman03}.  The DPLL procedure
gives rise to three non-trivial subsolvers: $\UP+\PL$ (unit
propagation and pure literal elimination are available), $\UP$ (only
unit propagation is available), $\PL$ (only pure literal elimination
is available). We associate each subsolver with the class of CNF
formulas for which it determines the satisfiability (this is
well-defined, since unit propagation and pure literal elimination are
confluent operations). Since the subsolvers clearly run in polynomial
time, $\UP+\PL$, $\UP$, and $\PL$ form base classes. We write
$\SUBSOLVER=\set{\UP+\PL, \UP, \PL}$.

\subsection{Miscellaneous Base Classes}

\paragraph{Renamable Horn} Let $X$ be a set of variables and $F$ a CNF
formula. We let $r_X(F)$ denote the CNF formula obtained from $F$ by
replacing for every variable $x\in X$, all occurrences of $x^\epsilon$
in $F$ with $x^{1-\epsilon}$, for $\epsilon\in \set{0,1}$. We call
$r_X(F)$ a \emph{renaming} of~$F$. Clearly $F$ and $r_X(F)$ are
equisatisfiable.  A CNF formula is called \emph{renamable Horn} if it
has a renaming which is Horn, and we denote the class of renamable Horn
formulas as $\RHORN$.  It is easy to see that $\HORN$ is a strict subset
of $\RHORN$.  One can find in polynomial time a Horn renaming of a given
CNF formula, if it exists~\cite{Lewis78}. Hence $\RHORN$ is a further
base class. In contrast to $\HORN$, $\RHORN$ is not clause-defined.

\paragraph{Forests}
Many NP-hard problems can be solved in polynomial time for problem
instances that are in a certain sense acyclic. The satisfiability
problem is no exception. There are various ways of defining a CNF
formula to be acyclic.  Here we consider acyclicity based on (undirected) incidence graphs: the \emph{incidence graph}
of a CNF formula $F$ is the bipartite graph whose vertices are the
variables and the clauses of $F$; a variable $x$ and a clause $C$ are
joined by an edge \myiff $x\in \var(C)$.
Let $\FOREST$ denote the class of CNF formulas whose undirected
incidence graphs are forests.  It is well known that $\FOREST$ forms
islands of tractability: the satisfiability of CNF
formulas whose incidence graphs have bounded treewidth can be decided in
linear time~\cite{FischerMakowskyRavve06,SamerSzeider10}. $\FOREST$ is
the special case of formulas with treewidth at most $1$.

\paragraph{Clusters}
A CNF formula $F$ is called a \emph{hitting} if any two distinct
clauses clash. Two clauses $C,C'\in F$ \emph{clash} if
they contain a complimentary pair of literals, \ie, $C\cap \ol{C'}\neq
\emptyset$.  A CNF formula is called a
\emph{clustering formula} if it is a variable disjoint union of hitting
formulas.  We denote by $\CLU$ the class of clustering
formulas. Clustering formulas not only allow polynomial-time SAT
decision, one can even count the number of satisfying truth assignments
in polynomial time. This is due to the fact that each truth assignment
invalidates at most one clause of a hitting
formula~\cite{Iwama89,NishimuraRagdeSzeider07}.


\begin{table}[tb]
\begin{threeparttable}
\newcommand{\scite}[1]{${}^\text{{\protect\cite{#1}}}$}



\begin{tabular*}{\textwidth}{@{\extracolsep{\fill}} *{4}{l} }
  \toprule
  \emph{Base Class}      & \textsc{Weak} & \textsc{Strong} &
  \textsc{Deletion}  \\
  \midrule
  $\CCC\in\SCHAEFER$
  & W[2]-h \scite{NishimuraRagdeSzeider04-informal}  (FPT)
  & FPT \scite{NishimuraRagdeSzeider04-informal}   
  & FPT \scite{NishimuraRagdeSzeider04-informal}       \\
  $\CCC\in\SUBSOLVER$    
  & W[P]-c \scite{Szeider05d}
  & W[P]-c \scite{Szeider05d} 
  & n/a       \\ 
  $\FOREST$ 
  & W[2]-h \scite{GaspersSzeider11a}  (FPT \scite{GaspersSzeider11a})
  & ?${}^\dagger$\ (?) 
  & FPT\\
  $\RHORN$ 
  & W[2]-h 
  & W[2]-h (?) 
  & FPT \scite{RazgonOSullivan08}\\
  $\CLU$  
  & W[2]-h \scite{NishimuraRagdeSzeider07} (FPT)  
  & W[2]-h \scite{NishimuraRagdeSzeider07} (FPT \scite{NishimuraRagdeSzeider07}) 
  & FPT \scite{NishimuraRagdeSzeider07}      \\
  \bottomrule
\end{tabular*}
\begin{tablenotes}\footnotesize 
\item[( )] It is indicated in parentheses if the complexity of the
  problem for 3CNF formulas is different from general CNF or unknown.
\item[?] It is open whether the problem is fixed-parameter tractable.
\item[$\dagger$] Theorem \ref{the:strong-tree} gives an fpt
  approximation for this problem.
\item [n/a] Deletion backdoor sets are undefined for base classes that
  are not clause-induced.
\end{tablenotes}

 \medskip
 \caption{The parameterized complexity of \textsc{Weak,
     Strong}, and \textsc{Deletion $\CCC$\hy  Backdoor Set Detection} for
 various base classes $\CCC$.}  
  \label{tab:SAT}
\end{threeparttable}
\end{table}

\section{Detecting Weak Backdoor Sets}
It turns out that for all base classes $\CCC$ considered in
this survey, \textsc{Weak $\CCC$-Backdoor Set Detection} is
$\W[2]$-hard. In several cases, restricting the input formula to 3CNF
helps, and makes \textsc{Weak $\CCC$-Backdoor Set Detection}
fixed-parameter tractable.

In the proof of the following proposition we use a general approach that
entails previously published proofs (such as in
\cite{GaspersSzeider11a,NishimuraRagdeSzeider04-informal,NishimuraRagdeSzeider07}) as
special cases.
\begin{proposition}\label{pro:weak-hard}\sloppypar
  \textsc{Weak $\CCC$-Backdoor Set Detection} is $\W[2]$-hard for all
  base classes $\CCC\in \SCHAEFER \cup \set{\RHORN, \FOREST, \CLU}$.
\end{proposition}
\begin{proof}
  We show $\W[2]$-hardness for $\CCC\in \{\TWOCNF$, $\HORN$, $\ZEROV$,
  $\RHORN$, $\FOREST$, $\CLU\}$. The hardness proofs for the remaining two classes
  $\ONEV$ and $\HORN^-$ are symmetric to the proofs for $\ZEROV$ and
  $\HORN$, respectively.

  Let $G$ be a CNF formula with a set $X\subseteq \var(G)$ of its
  variables marked as \emph{external}, all other variables of $G$ are
  called \emph{internal}.  We call $G$ an \emph{or-gadget} for a
  base class~$\CCC$ if $G$ has the following properties:
  \begin{enumerate}
  \item\label{weak-or:B} $G\notin \CCC$.
  \item\label{weak-or:A} $G\in \ONEV$.
  \item\label{weak-or:C} $G[x=1]\in \CCC$ holds for all $x\in X$. 
  \item\label{weak-or:D} For each clause $C\in G$
    either $X\subseteq C$ or $\var(C)\cap X=\emptyset$.
  \item\label{weak-or:E} $\var(G)\setminus X\neq \emptyset$.
  \item\label{weak-or:F} $G$ can be constructed in time
    polynomial in $\Card{X}$.
  \end{enumerate}
First, we show the following meta-result, and then we define or-gadgets
for the different base-classes.

  \emph{Claim 1:} If $\CCC$ is clause-induced,  closed under
  variable-disjoint union, and has an or-gadget for any number $\geq
  1$ of external variables, then \textsc{Weak $\CCC$-Backdoor Set
    Detection} is $\W[2]$-hard.
  
  We prove the claim by giving a parameterized reduction from the
  $\W[2]$\hy complete problem \textsc{Hitting Set
    (HS)}~\cite{DowneyFellows99}. Let $(\SSS,k)$,
  $\SSS=\set{S_1,\dots,S_m}$, be an instance of HS.
  Let $I=\set{1,\dots,m} \times \set{1,\dots,k+1}$.  For each $S_i$ we
  construct $k+1$ or-gadgets $G_i^1,\dots,G_i^{k+1}$ whose external
  variables are exactly the elements of $S_i$, and whose internal
  variables do not appear in any of the other gadgets $G_{i'}^{j'}$ for
  $(i',j')\in I\setminus \{(i,j)\}$.
  Let $F=\bigcup_{(i,j)\in I} G_i^j$.  From Property~\ref{weak-or:F} it
  follows that $F$ can be constructed from $\SSS$ in polynomial time.
  We show that $\SSS$ has a hitting set of size $k$ \myiff $F$
  has a weak $\CCC$-backdoor set of size~$k$.
  
  Assume $B \subseteq \bigcup_{i=1}^m S_i$ is a hitting set of $\SSS$ of
  size $k$. Let $\tau \in 2^B$ the truth assignment that sets all
  variables from $B$ to~1.  By Properties~\ref{weak-or:A} and
  \ref{weak-or:C}, $G_i^j[\tau]$ is satisfiable and belongs to $\CCC$
  for each $(i,j)\in I$.  By Property~\ref{weak-or:D},
  $\var(G_i^j[\tau]) \cap \var(G_{i'}^{j'}[\tau])=\emptyset$ for any two
  distinct pairs $(i,j),(i',j')\in I$. Consequently $F[\tau]$ is
  satisfiable, and since $\CCC$ is closed under variable-disjoint union,
  $F[\tau]$ belongs to~$\CCC$. Thus $B$ is a weak $\CCC$\hy backdoor set
  of $F$ of size~$k$.

  Conversely, assume that $B\subseteq \var(F)$ is a weak $\CCC$\hy
  backdoor set of $F$ of size $k$. Hence, there exists a truth
  assignment $\tau\in 2^B$ such that $F[\tau]$ is satisfiable and
  belongs to~$\CCC$.  Clearly for each $(i,j)\in I$, $G_i^j[\tau]$ is
  satisfiable (since $G_i^j[\tau]\subseteq F$), and $G_i^j[\tau]\in
  \CCC$ (since $\CCC$ is clause-induced).  However, since $G_i^j \notin
  \CCC$ by Property~\ref{weak-or:B}, $B\cap \var(G_i^j) \neq \emptyset$
  for each $(i,j)\in I$.  Let $1\leq i \leq m$.  By construction, $F$
  contains $k+1$ copies $G_i^1\dots,G_i^{k+1}$ of the same gadget. From
  Property~\ref{weak-or:E} it follows that all the $k+1$ copies are
  different.  Since $\Card{B}\leq k$, there must be some $x_i \in B$
  such that there are $1\leq j' < j'' \leq k+1$ with $x_i \in
  \var(G_i^{j'}) \cap \var(G_i^{j''})$. It follows that $x_i$ is an
  external variable of $G_i^{j'}$, hence $x_i \in B\cap
  S_i$. Consequently, $B$ is a hitting set of $\SSS$.

  Hence we have indeed a parameterized reduction from HS to \textsc{Weak
    $\CCC$-Backdoor Set Detection}, and Claim~1 is shown true.  We
  define for each class $\CCC\in \{\TWOCNF$, $\HORN$, $\ZEROV$,
  $\RHORN$, $\FOREST$, $\CLU\}$  an  or-gadget $F(\CCC)$ where
  $X=\{x_1,\dots,x_s\}$ is the set of external variables; internal
  variables are denoted~$z_i$.
  \begin{itemize}
  \item   $G(\TWOCNF)=\{X\cup \{z_1,z_2\}\}$.
  \item   $G(\HORN)=G(\ZEROV)=\{X\cup \{z_1\}\}$.
  \item   $G(\RHORN)=\{X\cup \{\neg z_1, \neg z_2\}, \{z_1,\neg z_2\}, \{\neg z_1,z_2\},
    \{z_1,z_2\} \}$.
  \item   $G(\FOREST)=\{X\cup \{\neg z_1, \neg z_2\}, \{ z_1, z_2\}\}$.
  \item   $G(\CLU)=\{X\cup \{z_1\}, \{z_1\}\}$. 
  \end{itemize}
  Since the considered classes $\CCC$ are clearly clause-induced and
  closed under variable-disjoint union, the proposition now follows from
  Claim~1.
\end{proof}

For base classes based on subsolvers, weak backdoor set detection is even
$\W[P]$\hy hard. This is not surprising, since the subsolvers allow a
propagation through the formula which is similar to the propagation in
problems like \textsc{Minimum Axiom Set} or \textsc{Degree 3 Subgraph
  Annihilator}~\cite{DowneyFellows99}. The proof of the following
theorem is based on a reduction from the $\W[P]$\hy complete problem
\textsc{Cyclic Monotone Circuit Activation}.
\begin{theorem}[\cite{Szeider05d}]\label{the:subsolver-weak}
  \textsc{Weak $\CCC$-Backdoor Set Detection} is $\W[P]$-complete for
  all base classes  $\CCC\in \SUBSOLVER$. This even holds if the
  input formula is in 3CNF.
\end{theorem}

In summary, we conclude that \textsc{Weak $\CCC$-Backdoor Set Detection}
is at least $\W[2]$\hy hard for all considered base classes.  If we restrict
our scope to 3CNF formulas, we obtain mixed results.

\begin{proposition}\label{prop:fpt-clause-defined-3}
  For every clause-defined class $\CCC$, \textsc{Weak $\CCC$-Backdoor
    Set Detection} is fixed-parameter tractable for input formulas in 3CNF.
\end{proposition}
\begin{proof}
  The result follows by a standard bounded search tree argument,
  sketched as follows.  Assume we are given a CNF formula $F\notin \CCC$
  and an integer $k$. We want to decide whether $F$ has a weak
  $\CCC$\hy backdoor set of size $\leq k$. Since $\CCC$ is
  clause-defined, $F$ contains a clause $C$ such that $\{C\}\notin
  \CCC$. Hence some variable of $\var(C)$ must belong to any weak
  $\CCC$\hy backdoor set of $F$. There are at most $3$ such variables,
  each of which can be set to true or to false. Hence we branch in at
  most 6 cases. By iterating this case distinction we build a search
  tree $T$, where each node $t$ of $T$ corresponds to a partial truth
  assignment $\tau_t$. We can stop building the tree at nodes of depth
  $k$ and at nodes $t$ where $F[\tau_t]\in \CCC$. It is now easy to see
  that $F$ has a weak $\CCC$\hy backdoor set of size at most $k$ if and
  only if $T$ has a leaf $t$ such that $F[\tau_t]\in \CCC$ and
  $F[\tau_t]$ is satisfiable. For each leaf we can check in polynomial
  time whether these properties hold.
\end{proof}

In particular, \textsc{Weak $\CCC$-Backdoor Set Detection} is
fixed-parameter tractable for $\CCC \in \SCHAEFER$ if the input
formula is in 3CNF.

The proof of Proposition~\ref{prop:fpt-clause-defined-3} can be extended
to the class $\CLU$ of clustering formulas.  Nishimura et
al.~\cite{NishimuraRagdeSzeider07} have shown that a CNF formula is a
clustering formula \myiff it does not contain (i) two clauses
$C_1,C_2$ that overlap ($C_1\cap C_2 \neq \emptyset$) but do not clash
($C_1\cap \ol{C_2}=\emptyset$), or (ii)~three clauses $D_1,D_2,D_3$
where $D_1$ and $D_2$ clash, $D_2$ and $D_3$ clash, but $D_1$ and $D_3$
do not clash.  $\{C_1,C_2\}$ is called an overlap obstruction,
$\{D_1,D_2,D_2\}$ is called a clash obstruction.  Each weak $\CLU$\hy
backdoor set of a CNF formula $F$ must contain at least one variable
from each overlap and each clash obstruction.  However, if $F$ is a 3CNF
formula, the number of variables of an overlap obstruction is at most 5,
and the number of variables of a clash obstruction is at most 7. Hence
we can find a weak $\CLU$\hy backdoor set of size at most $k$ with a
bounded search tree, which gives the following result.

\begin{proposition}\label{pro:clu}
  \textsc{Weak $\CLU$-Backdoor Set Detection} is fixed-parameter
  trac\-ta\-ble for 3CNF formulas.
\end{proposition}

\begin{proposition}
\textsc{Weak $\RHORN$-Backdoor Set
    Detection} is $\W[2]$\hy hard, even for 3CNF formulas.
\end{proposition}
%
%
%
\begin{proof}
  Similarly to the proof of Proposition~\ref{pro:weak-hard} we reduce
  from HS.  As gadgets we use formulas of the form $G=\{\{z_1, \neg x_1,\neg z_2\}$,
  $\{z_2, \neg x_2,\neg z_3\}, \dots, \{z_s, \neg x_s,\neg z_{s+1}\}$,
  $\{\neg z_{1}, z_{s+1}\}$, $\{\neg z_{1}, \neg z_{s+1}\}$,
  $\{z_{1}, z_{s+1}\}\}$, where $x_1,\dots,x_s$ are external
  variables and $z_1,\dots,z_{s+1}$ are internal variables. $G$ can be
  considered as being obtained form the complete formula $\SB \{
  z_{1}^\epsilon, z_{s+1}^\delta \} \SM \epsilon,\delta \in \{0,1\} \SE$
  by ``subdividing'' the clause $\{z_1, \neg z_{s+1} \}$. $G \notin \RHORN$
  but $G[x_i=0]\in \RHORN$. In fact, $r_X(G[x_i=0])\in \HORN$ for
  $X=\{z_{i+1},\dots,z_{s+1}\}$, hence no external variable needs to be
  renamed. Moreover, we can satisfy $G[x_i=0]$ by setting all external
  variables and $z_1$ to~0, and by setting $z_{s+1}$ to~1.

  Let $(\SSS,k)$, $\SSS=\{S_1,\dots,S_m\}$, be an instance of HS.  For
  each $S_i$ we construct $k+1$ gadgets $G_i^1,\dots,G_i^{k+1}$, each
  having $S_i$ as the set of its external variables, and the internal
  variables are new variables only used inside a gadget. We let $F$ to
  be the union of all such gadgets $G_i^j$ for $1\leq i \leq m$ and
  $1\leq j \leq k+1$.

  Similar to the proof of Proposition~\ref{pro:weak-hard} we can easily
  show that $\SSS$ has a hitting set of size $k$ \myiff $F$ has
  a weak $\RHORN$\hy backdoor set of size $k$. The proposition follows.
\end{proof}

According to Propositions~\ref{prop:fpt-clause-defined-3} and
\ref{pro:clu}, \textsc{Weak $\CCC$-Backdoor Set
  Detection} is fixed-parameter tractable for certain base classes
$\CCC$ and  input formulas in 3CNF. For the classes $\CCC$ covered
by Propositions~\ref{prop:fpt-clause-defined-3} and
\ref{pro:clu} it holds that  for every 3CNF formula $F\notin \CCC$ we can
find a set of variables of bounded size, an ``obstruction'', from
which at least one variable must be in any weak $\CCC$\hy backdoor
set of $F$.  Hence a weak $\CCC$ backdoor set of size at most $k$ can
be found by means of a bounded search tree algorithm.  The next result
shows that fixed-parameter tractability also prevails for the base
class $\FOREST$. However, the algorithm is considerably more
complicated, as in this case we do not have obstructions of bounded
size.

\begin{theorem}[\cite{GaspersSzeider11a}]\label{the:weak-forest}
 \textsc{Weak $\FOREST$-Backdoor Set Detection} is fixed-parameter
  tractable for 3CNF formulas.
\end{theorem}
\begin{proof}[Sketch] 
  We sketch the fpt algorithm from \cite{GaspersSzeider11a} deciding
  whether a 3CNF formula has a weak $\FOREST$\hy backdoor set of size
  $k$. We refer to \cite{GaspersSzeider11a} for the full details and the
  correctness proof.  Let $G$ denote the incidence graph of $F$.  The
  first step of the algorithm runs an fpt algorithm (with parameter
  $k'$) by Bodlaender \cite{Bodlaender94} that either finds $k' = 2k+1$
  vertex-disjoint cycles in $G$ or a feedback vertex set of $G$ of size
  at most $12k'^2-27k'+15$.
 
  In case a feedback vertex set $X$ is returned, a tree decomposition of
  $G \setminus X$ of width~1 is computed and $X$ is added to each bag of
  this tree decomposition.  As the \textsc{Weak $\FOREST$-Backdoor Set
    Detection} problem can be defined in Monadic Second Order Logic, a
  meta-theorem by Courcelle \cite{Courcelle90} can use this tree
  decomposition to conclude.
 
  In case Bodlaender's algorithm returns $k'$ vertex-disjoint cycles,
  the algorithm finds a set $S^*$ of $O(4^k k^6)$
  variables such that any weak $\FOREST$-backdoor set of size $k$
  contains at least one variable from $S^*$.  In this case, the algorithm
  recurses by considering all possibilities of assigning a value to a
  variable from $S^*$.
 
  Let $C_1, \ldots, C_{k'}$ denote the variable-disjoint cycles returned
  by Bodlaender's algorithm.  Consider a variable $x \in \var(F)$ and a
  cycle $C$. We say that $x$ \emph{kills} $C$ \emph{internally} if $x\in
  C$. We say that $x$ \emph{kills} $C$ \emph{externally} if $x \notin C$
  and $C$ contains a clause $u \in F$
  such that $x\in \var(u)$.
 
  As our $k'$ cycles are all vertex-disjoint, at most $k$ cycles may be
  killed internally. The algorithm goes through all choices of $k$
  cycles among $C_1, \ldots, C_{k'}$ that may be killed internally.  All
  other cycles, say $C_1, \ldots, C_{k+1}$, are not killed internally
  and need to be killed externally. The algorithm now computes a set $S\subseteq \var(F)$
  of size $O(k^6)$
  such that any weak $\FOREST$-backdoor set of size $k$, which is a subset
  of $\var(F) \setminus \bigcup_{i=1}^{k+1} \var(C_i)$,
  contains at least one variable from $S$. The union of all such $S$,
  taken over all choices of cycles to be killed internally, forms then
  the set $S^*$ that was to be computed.

  For each cycle from $C_1, \ldots, C_{k+1}$, compute
  its set of external killers in $\var(F) \setminus \bigcup_{i=1}^{k+1} \var(C_i)$.
  Only these external killers are considered from now on.
  If one such cycle
  has no such external killer, then there is no solution with the current
  specifications and the algorithm backtracks.  For each $i, 1\le i\le
  k+1$, let $x_i$ denote an external killer of $C_i$ with a maximum
  number of neighbors in $C_i$. The algorithm executes the first
  applicable from the following rules.
  \begin{description}
  \item[Multi-Killer Unsupported] If there is an index $i, 1\le i\le
    k+1$ such that $x_i$ has $\ell \ge 4k$ neighbors in $C_i$ and at
    most $4k^2+k$ external killers of $C_i$ have at least $\ell/(2k)$
    neighbors in $C_i$, then include all these external killers in $S$.
  \item[Multi-Killer Supported] If there is an index $i, 1\le i\le k+1$
    such that $x_i$ has $\ell \ge 4k$ neighbors in $C_i$ and more than
    $4k^2+k$ external killers of $C_i$ have at least $\ell/(2k)$ neighbors
    in $C_i$, then set $S = \set{x_i}$.
  \item[Large Overlap] If there are two cycles $C_i,C_j,1\le i \neq j
    \le k+1,$ with at least $16k^4+k$ common external killers, then
    set $S = \emptyset$.
  \item[Small Overlap] Otherwise, include in $S$ all vertices that are
    common external killers of at least two cycles from $C_1,\ldots,
    C_{k+1}$.
\end{description}
The algorithm recursively checks for each $s\in S^*$ whether the
formulas $F[s=0]$ and $F[s=1]$ have a weak $\FOREST$-backdoor
set of size $k-1$ and returns \textsc{Yes} if any such recursive call
was successful and \textsc{No} otherwise.
\end{proof}

\section{Detecting Strong Backdoor Sets}

\begin{proposition}[\cite{NishimuraRagdeSzeider04-informal}]
  \sloppy \textsc{Strong $\CCC$-Backdoor Set Detection} is fixed-parameter
  tractable for every base class $\CCC\in \SCHAEFER$. For $\CCC\in
  \{\ZEROV,\ONEV\}$, the problem is even solvable in polynomial time.
\end{proposition}
\begin{proof}  
  \sloppypar Consider a CNF formula $F$.  Strong $\HORN$\hy backdoor
  sets of $F$ are exactly the vertex covers of the positive primal graph
  of $F$, whose vertex set is $\var(F)$, two variables are joined by an
  edge if they appear together positively in a clause. Strong
  $\HORN^-$\hy backdoor sets can be characterized symmetrically.  Strong
  $\TWOCNF$\hy backdoor sets of $F$ are exactly the hitting sets of the
  hypergraph whose vertex set is $\var(F)$ and whose hyperedges are all
  the subsets $e\subseteq \var(F)$ of size three such that $e\subseteq
  \var(C)$ for a clause $C\in F$. Thus \textsc{Strong $\CCC$-Backdoor
    Set Detection} for $\CCC\in \{\HORN$, $\HORN^-$, $\TWOCNF\}$ can be
  accomplished by fpt algorithms for \textsc{Vertex Cover} \cite{ChenKanjXia10} and
  \textsc{3-Hitting Set} \cite{Fernau10}.
  The smallest strong $\ONEV$-backdoor set of $F$ is exactly the union
  of $\var(C)$ for all negative clauses $C\in F$, the smallest strong
  $\ZEROV$-backdoor set of $F$ is exactly the union of $\var(C)$ for all
  positive clauses $C\in F$.
\end{proof}

\begin{proposition}
  \textsc{Strong $\RHORN$-Backdoor Set Detection} is $\W[2]$\hy hard.
\end{proposition}
\begin{proof}
  The proof uses a reduction from HS similar to the proof of
  Proposition~\ref{pro:weak-hard}.  An instance $(\SSS,k)$,
  $\SSS=\{S_1,\dots,S_m\}$, of HS is reduced to a formula $F$ which is
  the union of certain gadgets $G_i^j$ for $1\leq i \leq m$ and $1\leq j
  \leq k+1$.  Let $V=\bigcup_{i=1}^m S_i$.  A gadget $G_i^j$ contains
  the four clauses $S_i\cup \{z_1,z_2\}$, $\{z_1,\neg z_2\}$, $\{\neg
  z_1,z_2\}$, and $\ol{V} \cup \{\neg z_1,\neg z_2\}$, where $z_1,z_2$
  are internal variables that do not occur outside the gadget.  Let
  $B\subseteq V$ be a hitting set of $\SSS$ and let $\tau\in 2^B$.  If
  $\tau$ sets at least one variable to $0$, then $\tau$ removes from
  each gadget the only negative clause, hence $r_{\var(F)}(F[\tau])\in
  \HORN$. On the other hand, if $\tau$ sets all variables from $B$ to
  $1$, then it removes from each gadget the only positive clause ($B$ is
  a hitting set). Hence, $F[\tau]\in \HORN$ in this case. Consequently
  $B$ is a strong $\RHORN$\hy backdoor set of~$F$.  Conversely, assume
  $B$ is a strong $\RHORN$\hy backdoor set of~$F$. Let $\tau \in 2^B$ be
  the all-1-assignment. For the sake of contradiction, assume there is a set $S_i$
  such that $B\cap S_i=\emptyset$. Since $\Card{B}=k$, $B\cap
  \var(G_i^j)=\emptyset$ for some $1\leq j \leq k+1$.  Now $F[\tau]$
  contains the subset $G_i^j[\tau]=\{S_i\cup \{z_1,z_2\}, \{z_1,\neg
  z_2\}, \{\neg z_1,z_2\}, \{\neg z_1,\neg z_2\} \}$ which is not
  renamable Horn, hence $B$ is not a strong $\RHORN$\hy backdoor set of
  $F$, a contradiction. Hence $B$ is a hitting set of $\SSS$.
\end{proof}

It is not known whether \textsc{Strong $\FOREST$-Backdoor Set Detection}
is fixed-parameter tractable nor whether \textsc{Strong
  $\RHORN$-Backdoor Set Detection} is fixed-parameter tractable for 3CNF
formulas. For the former problem, however, we know at least an fpt
approximation~\cite{GaspersSzeider11a}; see
Theorem~\ref{the:strong-tree} below.

The following result is shown by a reduction from 
\textsc{Cyclic Monotone Circuit Activation}, similarly to
Theorem~\ref{the:subsolver-weak}.
\begin{theorem}[\cite{Szeider05d}]
  \textsc{Strong $\CCC$-Backdoor Set Detection} is $\W[P]$\hy complete
  for every base class  $\CCC\in \SUBSOLVER$,
  even for formulas in 3CNF.
\end{theorem}

The bounded search tree method outlined above for \textsc{Weak
  $\CLU$-Backdoor Set Detection} for 3CNF formulas can clearly be
adapted for strong backdoors. Hence we get the following result.
\begin{proposition}
  \textsc{Strong $\CLU$-Backdoor Set Detection} is fixed-parameter
  trac\-ta\-ble for 3CNF formulas.
\end{proposition}

\subsection{Empty Clause Detection}

Dilkina \etal \cite{DilkinaGomesSabharwal07} suggested
to strengthen the concept of strong backdoor sets by means of
\emph{empty clause detection}. Let $\EEE$ denote the class of all CNF
formulas that contain the empty clause. For a base class $\CCC$ we put
$\CCC^\nil=\CCC\cup \EEE$; we call $\CCC^\nil$ the base class obtained
from $\CCC$ by adding empty clause detection.  Formulas often have
much smaller strong $\CCC^\nil$\hy backdoor sets than strong $\CCC$\hy
backdoor sets \cite{DilkinaGomesSabharwal07}.  Dilkina \etal show
that, given a CNF formula $F$ and an integer $k$, determining whether
$F$ has a strong $\HORN^\nil$\hy backdoor set of size $k$, is both
$\NP$\hy hard and $\coNP$\hy hard (here $k$ is considered just as part of
the input and not as a parameter). Thus, the non-parameterized search
problem for strong $\HORN$\hy backdoor sets gets harder when empty
clause detection is added. 
It turns out that also the parameterized problem gets harder when empty
clause detection is added.

%
\begin{theorem}[\cite{Szeider08c}]\label{the:empty}
  For every clause-induced base class $\CCC$ such that at
  least one satisfiable CNF formula does not belong to $\CCC$ the problem \textsc{strong
    $\CCC^\nil$\hy backdoor set} is $\W[1]$\hy hard.
\end{theorem}

The theorem clearly applies to all base classes in $\SCHAEFER \cup
\{\RHORN, \FOREST\}$. The proof from \cite{Szeider08c} relies on a reduction
from \cite{FellowsSzeiderWrightson06}, where a reduction to 3CNF formulas is
also given. Thus, Theorem \ref{the:empty} also holds for 3CNF formulas.

\section{Detecting Deletion Backdoor Sets}


In this section we consider the parameterized complexity of
\textsc{Deletion $\CCC$\hy Backdoor Set Detection} for the various base
classes $\CCC$ from above.  For most of the classes the complexity is easily
established as follows.  For Schaefer classes, strong and deletion
backdoor sets coincide, hence the FPT results carry over.  The subsolver
classes are not clause-induced, hence it does not make sense to consider
deletion backdoor sets.  \textsc{Deletion $\FOREST$-Backdoor Set
  Detection} can be solved by algorithms for a slight variation of the feedback vertex set
problem, and is therefore FPT. One has only to make sure that the feedback
vertex set contains only variables and no clauses. This, however, can
be achieved by using algorithms for \textsc{Weighted Feedback Vertex
  Set} \cite{RamanSaurabhSubramanian06,ChenFominLiuLuVillanger08}.

It is tempting to use Chen \etal's FPT algorithm for directed
feedback vertex set \cite{ChenLiuLuOsullivanRazgon08} for the
detection of deletion backdoor sets. The corresponding base class would
contain all CNF formulas with acyclic \emph{directed} incidence graphs (the
orientation of edges indicate whether a variable occurs positively or
negatively).  Unfortunately this class is not suited as a base class
since it contains formulas where each clause contains either only
positive literals or only negative literals, and SAT is well known to
be NP-hard for such formulas~\cite{GareyJohnson79}.

Hence we are left with the classes $\CLU$ and $\RHORN$. 

For the detection of deletion $\CLU$-backdoor sets we can use overlap
obstructions and clash obstructions, as defined before Proposition
\ref{pro:clu}. With each obstruction, we associate a deletion pair which
is a pair of sets of variables. With an overlap obstruction
$\{C_1,C_2\}$, we associate the deletion pair
\begin{align*}
\{ \var(C_1 \cap C_2), \var((C_1 \setminus C_2) \cup (C_2 \setminus C_1) ) \},
\end{align*}
and with a clash obstruction $\{D_1,D_2,D_3\}$, we associate the deletion pair
\begin{align*}
\{ \var((D_1 \setminus D_3) \cap \overline{D_2}), \var((D_3 \setminus D_1) \cap \overline{D_2}) \}.
\end{align*}
For a formula $F$, let $G_F$ denote the graph with vertex set $\var(F)$ that has an edge $xy$
\myiff there is a deletion pair $\{X,Y\}$ of $F$ with $x\in X$ and $y\in Y$. 
Nishimura \etal \cite{NishimuraRagdeSzeider07} have shown that a set $X \subseteq \var(F)$
is a deletion $\CLU$-backdoor set of $F$ \myiff $X$ is a vertex cover of $G_F$.
Thus, the detection of a deletion $\CLU$-backdoor set of size $k$ can be reduced to the problem of
checking whether $G_F$ has a vertex cover of size $k$, for which there exist very fast algorithms
(see for example \cite{ChenKanjXia10}).

\begin{proposition}[\cite{NishimuraRagdeSzeider07}]\label{pro:clu-deletion}
  \textsc{Deletion $\CLU$-Backdoor Set Detection} is fixed-pa\-ra\-me\-ter
  tractable.
\end{proposition}

\sloppypar The remaining case is the class $\RHORN$. As noted by Gottlob and
Szeider~\cite{GottlobSzeider08} without proof (see also \cite{RazgonOSullivan09}),
one can show fixed-parameter tractability of
\textsc{Deletion $\RHORN$-Backdoor Set Detection} by reducing it to the
problem \textsc{2SAT Deletion}. The latter problem takes as input a 2CNF
formula and an integer $k$ (the parameter), and asks whether one can
make the formula satisfiable by deleting at most~$k$ clauses.
\textsc{2SAT Deletion}~was shown fixed-parameter tractable by Razgon and
O'Sullivan~\cite{RazgonOSullivan09}. Here we give the
above mentioned reduction.

\begin{lemma}\label{lem:rhorn-2satdel}
  There is a parameterized reduction from \textsc{Deletion
    $\RHORN$-Backdoor Set Detection} to \textsc{2SAT Deletion}.
\end{lemma}
\begin{proof}
  Let $(F,k)$ be a given instance of \textsc{Deletion $\RHORN$-Backdoor
    Set Detection}.  We construct a graph $G=(V,E)$ by taking as
  vertices all literals $x^\epsilon$, for $x\in \var(F)$ and $\epsilon
  \in \{0,1\}$, and by adding two groups of edges. The first group
  consists of all edges $x^0,x^1$ for $x\in\var(F)$, the second group
  consists of all edges $x^\epsilon y^\delta$ for $x,y\in \var(F)$,
  $\epsilon,\delta \in \{0,1\}$, such that $x^\epsilon,y^\delta \in C$
  for some $C\in F$. Observe that the edges of the first group form a
  perfect matching $M$ of the graph $G$.
  
  \emph{Claim 1.} $F$ has a deletion $\RHORN$\hy backdoor set of size
    at most $k$ \myiff $G$ has a vertex cover with at most
    $\Card{M}+k$ vertices.
 
  ($\Rightarrow$) Let $B$ be a deletion $\RHORN$\hy backdoor set of $F$
  of size at most $k$ and $X \subseteq \var(F)\setminus B$ such that
  $r_X(F-B)\in \HORN$.  Let $N=\SB x^0 \SM x\in \var(F) \setminus X \SE
  \cup \SB x^1 \SM x\in X \SE$.  Let $K=\SB x^0,x^1 \SM x \in B \SE \cup
  N$. By definition, $\Card{K}=\Card{M}+\Card{B}\leq \Card{M}+k$. We
  show that $K$ is a vertex cover of $G$. Consider an edge $e=x^0x^1\in
  M$ of the first group.
  If $x\in X$, then $x^1 \in N\subseteq K$ and if $x\notin X$, then $x^0 \in N\subseteq K$.
  Hence $e$ is covered by $K$.
  It
  remains to consider an edge $f=x^\epsilon y^\delta$ of the second
  group. If $x\in B$ or $y\in B$, then this edge is covered by
  $K$. Hence assume $x,y \notin B$.  By construction of $G$, there
  is a clause $C\in F$ with $x^\epsilon,y^\delta \in C$. Since $x,y
  \notin B$, there is also a clause $C'\in F-B$ with
  $x^\epsilon,y^\delta \in C$.  Since $C'$ corresponds to a Horn clause
  $C'' \in r_X(F-B)$, at least one of the literals $x^\epsilon,y^\delta$
  belongs to $N$, and hence $K$ covers the edge~$f$.  Hence the first
  direction of Claim~1 follows.

  ($\Leftarrow$) Let $K$ be a vertex cover of $G$ with at most
  $\Card{M}+k$ vertices. Let $B\subseteq \var(F)$ be the set of all
  variables $x$ such that both $x^0,x^1\in K$. Clearly $\Card{B}\leq k$.
  Let $X\subseteq \var(F)\setminus B$ such that $x^1\in K$.  We show
  that $r_X(F-B) \in \HORN$.  Let $x^\delta,y^\epsilon$ be two literals
  that belong to a clause $C''$ of $r_X(F-B)$. We show that $\epsilon=0$
  or $\delta=0$. Let $C'\in F-B$ the clause that corresponds to $C''$,
  and let $x^{\epsilon'},y^{\delta'}\in C'$. It follows that
  $x^{\epsilon'}y^{\delta'}\in E$, and since $K$ is a vertex cover of
  $G$, $x^{\epsilon'}\in K$ or $y^{\delta'}\in K$.  If $x^{\epsilon'}\in
  K$ then $\epsilon=0$, if $y^{\delta'}\in K$ then $\delta=0$.  Since
  $x^\delta,y^\epsilon\in C''\in r_X(F-B)$ were chosen arbitrarily, we
  conclude that $r_X(F-B)\in \HORN$.  Hence Claim~1 is shown.

  Mishra \etal \cite{MishraRSSS07} already observed that a reduction
  from \cite{ChenKanj05} can be adapted to show that this above-guarantee
  vertex cover problem can be reduced to \textsc{2SAT Deletion}. For
  completeness, we give a reduction here as well.

  We construct a 2CNF formula $F_2$ from $G$.  For each vertex
  $x^\epsilon$ of $G$ we take a variable $x_\epsilon$.  For each edge
  $x^0x^1\in M$ we add a negative clause $\{\neg x_0, \neg x_1\}$, and
  for each edge $x^\epsilon y^\delta\in E\setminus M$ we add a positive
  clause $\{x_\epsilon,y_\delta\}$.

  \emph{Claim 2.} $G$ has a vertex cover with at most $\Card{M}+k$
    vertices \myiff we can delete at most $k$ negative clauses
    from $F_2$ to obtain a satisfiable formula.

  ($\Rightarrow$) Let $K$ be a vertex cover of $G$.  We delete from
  $F_2$ all negative clauses $\{\neg x_0, \neg x_1\}$ where both
  $x_0,x_1\in K$ (there are at most $k$ such clauses) and obtain a 2CNF
  formula $F_2'$.  We define a truth assignment $\tau \in
  2^{\var(F_2')}$ by setting a variable to 1 \myiff it belongs
  to $K$. It remains to show that $\tau$ satisfies $F_2'$.  The negative
  clauses are satisfied since $\tau$ sets exactly one literal of a
  negative clause $\{\neg x_0, \neg x_1\}\in F_2'$ to~1 and exactly one
  to~0.  The positive clauses are satisfied since each positive clause
  $\{x_\epsilon,y_\delta\}$ corresponds to an edge $x_\epsilon
  y_\delta\in E$, and since $K$ is a vertex cover, $\tau$ sets at least
  one of the variables $x_\epsilon,y_\delta$~to~1.

  ($\Leftarrow$) Let $F_2'$ be a satisfiable formula obtained from $F_2$
  by deleting at most $k$ negative clauses. Let $D = \SB x\in \var(F)
  \SM \{\neg x_0, \neg x_1\} \in F_2\setminus F_2' \SE$.  Let $\tau$ be
  a satisfying truth assignment of $F_2'$. We define a set $K$ of
  vertices of $G$ by setting $K=\SB x^0,x^1 \SM x\in D \SE \cup \SB
  x^{\tau(x)} \SM x \in \var(F)\setminus D \SE$, and we observe that
  $\Card{K}\leq \Card{M}+k$. It remains to show that $K$ is a vertex
  cover of $G$. Consider an edge $e=x^0x^1\in M$ of the first group. If
  $x\in D$ then $x^0,x^1\in K$; if $x\notin D$ then $x^{\tau(x)} \in K$,
  hence $e$ is covered by $K$. Now consider an edge $f=x^\epsilon
  y^\delta\in E\setminus M$ of the second group. If $x\in D$ or $y\in
  D$ then $f$ is clearly covered by $K$. Hence assume $x,y\notin D$.
  By definition, there is a positive clause $\{x_\epsilon,
  y_\delta\}\in F_2'\subseteq F_2$. Since $\tau$ satisfies $F_2'$, it
  follows that $\tau(x_\epsilon)=1$ or
  $\tau(y_\delta)=1$. Consequently $x^{\epsilon}\in K$ or
  $y^{\delta}\in K$, thus $K$ covers~$f$.
  Hence Claim~2 is shown. 

  Next we modify $F_2$ by replacing each positive clause
  $C=\{x_\epsilon,y_\delta\}$ with $2k+2$ ``mixed'' clauses
  $\{x_\epsilon,z_C^i\}, \{\neg z_C^i, y_\delta\}$, for $1\leq i \leq
  k+1$, where the $z_C^i$ are new variables. Let $F_2^*$ denote the 2CNF
  formula obtained this way from $F_2$. 

  \emph{Claim 3.} We can delete at most $k$ negative clauses from $F_2$
    to obtain a satisfiable formula \myiff we can delete at most
    $k$ clauses from $F_2^*$ to obtain a satisfiable formula.

  The claim follows easily from the following considerations.  We
  observe that each pair of mixed clauses $\{x_\epsilon,z_C^i\}, \{\neg
  z_C^i, y_\delta\}$ is semantically equivalent with
  $C=\{x_\epsilon,y_\delta\}$. Hence, if $F_2$ can be made satisfiable
  by deleting some of the negative clauses, we can also make $F_2^*$
  satisfiable by deleting the same clauses. However, deleting some of
  the mixed clauses does only help if we delete at least one from each
  of the $k+1$ pairs that correspond to the same clause $C$. Hence also
  Claim~3 is shown true.
  Claims 1--3 together establish the lemma.
\end{proof}

Razgon and O'Sullivan's result \cite{RazgonOSullivan09} together with
Lemma~\ref{lem:rhorn-2satdel} immediately give the following.
\begin{proposition}
\textsc{Deletion  $\RHORN$-Backdoor Set Detection}   is fixed-parameter
tractable.
\end{proposition}

\section{Permissive Problems}

We consider any function $p$ that assigns nonnegative integers to CNF
formulas as a \emph{satisfiability parameter}. 
In particular we are
interested in such satisfiability parameters~$p$ for which the 
following parameterized problem is fixed-parameter tractable:
\begin{quote}
  $\SAT(p)$
  
  \emph{Instance}: A CNF formula $F$ and an integer $k\geq 0$.

  \emph{Parameter}: The integer $k$.
  
  \emph{Task}:  Determine whether $F$ is satisfiable or determine
  that $p(F)>k$.
\end{quote}
Note that an algorithm that solves the problem has the freedom of
deciding the satisfiability of some formulas $F$ with $p(F)>k$, hence
the exact recognition of formulas $F$ with $p(F)\leq k$ can be avoided.
Thus $\SAT(p)$ is not a usual decision problem, as there are three
different outputs, not just two.  If $\SAT(p)$ is fixed-parameter
tractable then we call $p$ an fpt satisfiability parameter, and we say
that ``the satisfiability of CNF formulas of bounded $p$ is fixed-parameter
tractable'' (cf.~\cite{Szeider04b}). We write $\THREESAT(p)$ if the
input is restricted to 3CNF formulas.

Backdoor sets provide a generic way to define
satisfiability parameters.  Let $\CCC$ be a base class and $F$ a CNF
formula. We define $\wb_\CCC(F)$, $\sb_\CCC(F)$ and $\db_\CCC(F)$ as the size
of a smallest weak, strong, and deletion $\CCC$\hy backdoor set of $F$,
respectively. 

Of course, if the detection of the respective $\CCC$\hy backdoor set is
fixed-parameter tractable, then $\wb_\CCC$, $\sb_\CCC$, and $\db_\CCC$
are fpt satisfiability parameters.  
However, it is possible that
$\wb_\CCC$, $\sb_\CCC$, or $\db_\CCC$ are fpt satisfiability parameters
but the corresponding $\CCC$\hy backdoor set detection problem is
$\W[1]$\hy hard. The problems $\SAT(\wb_\CCC)$, $\SAT(\sb_\CCC)$, and
$\SAT(\db_\CCC)$ can therefore be considered as more ``permissive''
versions of the ``strict'' problems \textsc{Weak}, \textsc{Strong}, and
\textsc{Deletion $\CCC$\hy Backdoor Set Detection}, the latter require
to find a backdoor set even if the given formula is trivially seen to be
satisfiable or unsatisfiable.  The distinction between permissive and
strict versions of problems have been considered in a
related context by Marx and
Schlotter~\cite{MarxSchlotter10,MarxSchlotter10b} for parameterized
$k$-neighborhood local search.  Showing hardness for permissive problems
$\SAT(p)$ seems to be a much more difficult task than for the strict
problems. So far we could establish only few such hardness results.

\begin{proposition}
  $\SAT(\wb_\CCC)$  is $\W[1]$-hard for
  all $\CCC\in \SCHAEFER \cup \{\RHORN\}$.
\end{proposition}
\begin{proof}
  We will show a more general result, that $\W[1]$-hardness holds for
  all base classes that contain all anti-monotone 2CNF
  formulas. A CNF formula is \emph{anti-monotone} if all its clauses are negative.
  Let $\CCC$ be a base class that contains all anti-monotone 2CNF
  formulas.

  We show that $\SAT(\wb_\CCC)$ is $\W[1]$-hard by reducing from
  \textsc{Partitioned Clique}, also known as \textsc{Multicolored
    Clique}. This problem takes as input a $k$-partite graph and asks
  whether the graph has a clique on $k$ vertices. The integer $k$ is the
  parameter. The problem is well-known to be $\W[1]$\hy complete
  \cite{Pietrzak03}.

  Let $H=(V,E)$ with $V=\bigcup_{i=1}^k V_i$ be an instance of this
  problem.  We construct a CNF formula $F$ as follows.  We consider the
  vertices of $H$ as variables and add clauses $\{\neg u,\neg v\}$ for
  any two distinct vertices 
  such that $uv\notin E$.
  For each $1\leq i \leq k$, we add the clause $V_i$.
  This completes the construction of~$F$.

  \noindent We show that the following statements are equivalent:
  \begin{enumerate}
  \item[(1)] $F$ is satisfiable
  \item[(2)] $H$ contains a $k$\hy clique.
  \item[(3)] $F$ has a weak $\CCC$\hy backdoor set of size at most~$k$.
  \end{enumerate}
  
  (1)$\Rightarrow$(2).  Let $\tau$ be a satisfying assignment of $F$.
  Because of the clause $V_i$, $\tau$ sets at least one variable of $V_i$ to
  $1$, for each $1\leq i \leq k$. 
  As each $V_i$ is an independent set, $F$ contains a clause $\{\neg u,\neg v\}$
  for every two distinct vertices in $V_i$. Thus,
  $\tau$ sets exactly one variable of $V_i$ to $1$, for each
  $1\leq i \leq k$. The clauses of $F$ also imply
  that $v_iv_j\in E$ for each $1\leq i < j \leq k$, since
  otherwise $\tau$ would falsify the clause $\{\neg v_i, \neg
  v_j\}$. Hence $v_1,\dots,v_k$ induce a clique in $H$.

  (2)$\Rightarrow$(3).  Assume $v_1,\dots,v_k$ induce a clique in $H$,
  with $v_i\in V_i$. We show that $B=\{v_1,\dots,v_k\}$ is a weak
  $\CCC$\hy backdoor set of $F$.  Let $\tau\in 2^B$ be the truth
  assignment that sets all variables of $B$ to $1$. This satisfies
  all the clauses $V_i, 1\le i\le k$. Thus, $F[\tau]$ is an anti-monotone 2CNF
  formula. Therefore it is in $\CCC$ and it is satisfiable as it is 0-valid.
  Hence $B$ is a weak $\CCC$\hy backdoor set of $F$.

  (3)$\Rightarrow$(1). Any formula that has a weak backdoor set is
  satisfiable.

  Since all three statements are equivalent, we conclude that
  $\SAT(\wb_\CCC)$ is $\W[1]$-hard. This shows the proposition for
  the base classes $\HORN$, $\TWOCNF$, $\ZEROV$, and $\RHORN$, as
  they contain all anti-monotone 2CNF formulas.
  The hardness for $\HORN^-$ and $\ONEV$ follows by symmetric arguments
  from the hardness of $\HORN$ and $\ZEROV$, respectively.
\end{proof}

\bigskip

In general, if we have an \emph{fpt approximation algorithm}
\cite{CaiHuang10,ChenGroheGruebnerG06,DowneyFellowsMcCartin06} for a
strict backdoor set detection problem, then the corresponding permissive
problem $\SAT(p)$ is fixed-parameter tractable.  For instance, if we
have an fpt algorithm that, for a given pair $(F,k)$ either outputs a
weak, strong, or deletion $\CCC$\hy backdoor set of $F$ of size at most
$f(k)$ or decides that $F$ has no such backdoor set of size at most $k$,
then clearly $\wb_\CCC$, $\sb_\CCC$, and $\db_\CCC$, respectively, is an
fpt satisfiability parameter.

This line of
reasoning is used in the next theorem to show that $\sb_\FOREST$ is an
fpt satisfiability parameter. This result labels $\FOREST$ as the first
nontrivial base class $\CCC$ for which $\sb_\CCC$ is an fpt
satisfiability parameter and $\sb_\CCC\neq\db_\CCC$. Hence the
additional power of strong $\FOREST$\hy backdoor sets over deletion
$\FOREST$\hy backdoor sets is accessible.

\begin{theorem}[\cite{GaspersSzeider11a}]\label{the:strong-tree}
\sloppy \textsc{Strong $\FOREST$\hy Backdoor Set Detection} admits a $2^k$
   fpt-approximation.
   Hence $\SAT(\sb_\FOREST)$ is fixed-parameter tractable.
\end{theorem}
\begin{proof}[Sketch] 
  We sketch the fpt-approximation algorithm from
  \cite{GaspersSzeider11a} which either concludes that a CNF formula $F$
  has no strong $\FOREST$-backdoor set of size $k$ or returns one of
  size at most $2^k$.  We refer to \cite{GaspersSzeider11a} for the full
  details and the correctness proof.  Let~$G$ denote the incidence graph
  of $F$.  The first step of the algorithm runs, similarly to the proof
  of Theorem~\ref{the:weak-forest}, the fpt algorithm (with parameter
  $k'$) by Bodlaender \cite{Bodlaender94} that either finds $k' = k^2
  2^{k-1}+k+1$ vertex-disjoint cycles in $G$ or a feedback vertex set of
  $G$ of size at most $12k'^2-27k'+15$.
 
  In case a feedback vertex set $X$ is returned, a tree decomposition of
  $G \setminus X$ of width~1 is computed and $X$ is added to each bag of
  this tree decomposition.  As the \textsc{Strong $\FOREST$-Backdoor Set
    Detection} problem can be defined in Monadic Second Order Logic, a
  meta-theorem by Courcelle \cite{Courcelle90} can be used to decide the
  problem in linear time using this tree decomposition.

  In case Bodlaender's algorithm returns $k'$ vertex-disjoint cycles,
  the algorithm finds a set $S^*$ of $O(k^{2k}2^{k^2-k})$
  variables such that every strong $\FOREST$-backdoor set of size $k$
  contains at least one variable from $S^*$.  In this case, the algorithm
  recurses by considering all possibilities of including a variable from
  $S^*$ in the backdoor set.
 
 Let $C_1, \ldots, C_{k'}$ denote the variable-disjoint cycles returned by Bodlaender's algorithm.
 Consider a variable $x \in \var(F)$ and a cycle $C$. We say that $x$ \emph{kills} $C$ \emph{internally}
 if $x\in C$. We say that $x$ \emph{kills} $C$ \emph{externally} if $x\notin C$ and $C$ contains two clause $u,v \in F$
 such that $x\in u$ and $\neg x\in v$. We say in this case that $x$ kills $C$ externally in $u$ and $v$.

 The algorithm goes through all $\binom{k'}{k}$ ways to choose $k$ cycles among $C_1, \ldots, C_{k'}$ that may be killed internally.
 All other cycles, say $C_1, \ldots, C_{k''}$ with $k'' = k'-k$, are not killed internally. We refer to these cycles as $C''$-cycles.
 The algorithm now computes a set $S\subseteq \var(F)$
 of size at most $2$
 such that any strong $\FOREST$-backdoor set of size $k$, which is a subset
 of $\var(F) \setminus \bigcup_{i=1}^{k''} \var(C_i)$,
 contains at least one variable from $S$. The union of all such $S$,
 taken over all choices of cycles to be killed internally, forms then
 the set $S^*$ that was to be computed.

 From now on, consider only killers in $\var(F) \setminus \bigcup_{i=1}^{k''} \var(C_i)$.
 For each $C''$-cycle $C_i$, consider vertices $x_i,u_i,v_i$ such that $x_i$ kills $C_i$ externally in $u_i$ and $v_i$
 and there is a path $P_i$ from $u_i$ to $v_i$ along the cycle $C_i$ such that if any variable kills $C_i$ externally in two clauses $u_i'$ and $v_i'$
 such that $u_i',v_i'\in P_i$, then $\set{u_i,v_i} = \set{u_i',v_i'}$. Note that any variable that does not kill $C_i$ internally,
 but kills the cycle $Cx_i = P_i\cup \set{x_i}$ also kills the cycle $C_i$ externally in $u_i$ and $v_i$. We refer to such external killers
 as \emph{interesting}.
 
 The algorithm executes the first applicable from the following rules.
 
\begin{description}
 \item[No External Killer] If there is an index $i, 1\le i\le k''$, such that $Cx_i$ has no external killer, then set $S := \set{x_i}$.
\item[Killing Same Cycles] If there are variables $y$ and $z$ and at least $2^{k-1}+1$ $C''$-cycles such that both $y$ and $z$ are interesting external killers of
 each of these $C''$-cycles,
 then set $S := \set{y,z}$.
 \item[Killing Many Cycles] If there is a variable $y$ that is an interesting external killer of at least $k\cdot 2^{k-1}+1$ $C''$-cycles, then set $S:=\set{y}$.
 \item[Too Many Cycles] Otherwise, set $S= \emptyset$.
\end{description}
For each $s\in S^*$ the algorithm calls itself recursively to
compute a strong $\FOREST$-backdoor set for $F[s=0]$ and for $F[s=1]$
with parameter $k-1$. If both recursive calls return backdoor sets, the
union of these two backdoor sets and $\{s\}$ is a strong $\FOREST$-backdoor set for
$F$. It returns the smallest such backdoor set obtained for all choices
of $s$, or \textsc{No} if for each $s\in S^*$ at least one recursive call
returned \textsc{No}.
\end{proof}

\section{Comparison of Parameters}

Satisfiability parameters can be compared \mywrt their
generality. Let $p,q$ be satisfiability parameters.  We say that $p$ is
\emph{at least as general} as $q$, in symbols $p\preceq q$, if there
exists a function $f$ such that for every CNF formula $F$ we have
$p(F)\leq f(q(F))$. Clearly, if $p \preceq q$ and $\SAT(p)$ is fpt, then so
is $\SAT(q)$. If $p \preceq q$ but not $q \preceq p$, then $p$ is
\emph{more general} than $q$. If neither $p \preceq q$ nor $q \preceq p$ then
$p$ and $q$ are \emph{incomparable}.

As discussed above, each base class $\CCC$ gives rise to
three satisfiability parameters $\wb_\CCC(F)$, $\sb_\CCC(F)$ and
$\db_\CCC(F)$.  If $\CCC$ is clause-induced, then $\sb_\CCC \preceq
\db_\CCC$; and if $\CCC\subseteq \CCC'$, then 
$\sb_{\CCC'} \preceq \sb_{\CCC}$ and $\db_{\CCC'} \preceq \db_{\CCC}$.

By associating certain graphs with CNF formulas one can use graph
parameters to define satisfiability parameters. The most commonly used
graphs are the primal, dual, and incidence graphs.  The primal graph
of a CNF formula~$F$ has as vertices the variables of $F$, and two
variables are adjacent if they appear together in a clause. The dual
graph has as vertices the clauses of $F$, and two clauses $C,C'$ are
adjacent if they have a variable in common (\ie, if $\var(C)\cap
\var(C')\neq \emptyset$). The incidence graph, as already defined
above, is a bipartite graph, having as vertices the variables and the
clauses of $F$; a variable $x$ and a clause $C$ are adjacent if $x\in
\var(C)$.  The directed incidence graph is obtained from the incidence
graph by directing an edge $xC$ from $x$ to $C$ if $x\in C$ and from
$C$ to $x$ if $\neg x \in C$.

The treewidth of the primal, dual, and incidence graph gives fpt
satisfiability parameters, respectively. The treewidth of the incidence
graph is more general than the other two satisfiability
parameters~\cite{KolaitisVardi00}.  The clique-width of the three graphs
provides three more general satisfiability parameters. However, these
satisfiability parameters are unlikely fpt: It is easy to see that SAT
remains NP-hard for CNF formulas whose primal graphs are cliques, and
for CNF formulas whose dual graphs are cliques. Moreover, $\SAT$,
parameterized by the clique-width of the incidence graph is $\W[1]$\hy
hard, even if a decomposition is provided~\cite{OrdyniakPaulusmaSzeider10}.
However, the clique-width of
directed incidence graphs is an fpt satisfiability parameter which is
more general than the treewidth of incidence
graphs~\cite{CourcelleMakowskyRotics01,FischerMakowskyRavve06}.

How do fpt satisfiability parameters based on decompositions and fpt
satisfiability parameters based on backdoor sets compare to each other?

Each base class $\CCC$ considered above, except for  the class $\FOREST$,  contains CNF formulas whose directed
incidence graphs have arbitrarily large clique-width. Hence none of the
decomposition based parameters is at least as general as the parameters
$\sb_\CCC$ and $\db_\CCC$. On the other hand, taking the disjoint union of
$n$ copies of a CNF formula multiplies the size of backdoor sets by $n$
but does not increase the width. Hence no backdoor based parameter is
more general than decomposition based parameters.

Thus, almost all considered backdoor based fpt satisfiability
parameters are incomparable with almost all considered decomposition
based fpt satisfiability parameters. A notable exception is the
satisfiability parameter $\db_\FOREST$. It is easy
to see that the treewidth of the incidence graph of a CNF formula is
no greater than the size of a smallest deletion $\FOREST$\hy
backdoor set plus one, as the latter forms a feedback vertex set of the
incidence graph. Thus the treewidth of incidence graphs is a more general
satisfiability parameter than the size of a smallest deletion
$\FOREST$\hy backdoor sets.  However, one can construct CNF formulas
$F$ with $\sb_\FOREST(F)=1$ whose directed incidence graph has
arbitrarily large clique-width. Just take a formula whose incidence
graph is a subdivision of a large square grid, and add a further
variable $x$ such that on each path which is a subdivision of one edge
of the grid there is a clause containing $x$ and a clause containing
$\neg x$. Thus, the satisfiability parameter $\sb_\FOREST$, which is
fpt by Theorem~\ref{the:strong-tree}, is incomparable to all the
decomposition based satisfiability parameters considered above.

Figure~\ref{fig:diagram-parms} shows the relationship between some of
the discussed fpt satisfiability parameters.

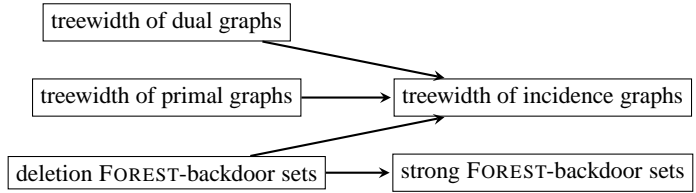
\begin{figure}
  \centering

   \begin{tikzpicture}

     \draw 
     (5,0)      node[draw] (twinc) {treewidth of incidence graphs} 
     (0,0)      node[draw] (twpri) {treewidth of primal  graphs} 
     (0,1)      node[draw] (twdua) {treewidth of dual graphs} 
     (0,-1)      node[draw] (deltree) {deletion $\FOREST$\hy backdoor sets}
     (5,-1)      node[draw] (strtree) {strong $\FOREST$\hy backdoor sets};

    \draw[thick,->,shorten >=1pt,>=stealth]    
    (twpri) edge (twinc)
    (twdua) edge (twinc)
    (deltree) edge (twinc)
    (deltree) edge (strtree);

   \end{tikzpicture}%

  \caption{Relationship between some fpt satisfiability parameters. An
    arrow from $A$ to $B$ means that $B$ is more general than $A$. If
    there is now arrow between $A$ and $B$ then $A$ and $B$ are
    incomparable.}
  \label{fig:diagram-parms}
\end{figure}

\section{Kernels}

The use of strong or deletion backdoor sets for SAT decision, \mywrt
a base class $\CCC$, involves two tasks:

\begin{enumerate}
\item \emph{backdoor detection}, to find a strong (or deletion) backdoor set
  of size at most $k$, or to report that such a backdoor set does not
  exist,
\item \emph{backdoor evaluation}, to use a given strong (or deletion) backdoor set of
  size at most $k$ to determine whether the CNF formula under
  consideration is satisfiable.
\end{enumerate}

In each case where backdoor detection is fixed-parameter tractable, one
can now ask whether the detection problem admits a polynomial
kernel. For instance, for the classes $\HORN$ and $\TWOCNF$, backdoor
detection can be rephrased as \textsc{Vertex Cover} or as
\textsc{3-Hitting Set} problems, as discussed above, and therefore
admits polynomial kernels \cite{ChenKanjJia01,Abukhzam10}.

Backdoor evaluation is trivially fixed-parameter tractable for any
base class, but it is unlikely that it admits a polynomial kernel.

\begin{proposition}[\cite{Szeider11a}]
  \textsc{$\CCC$\hy Backdoor Set Evaluation} does not admit a
  polynomial kernel for any self-reducible base class $\CCC$ unless
  $\NP \subseteq \coNP/\text{\normalfont poly}$.
\end{proposition}
This proposition is a trivial consequence of the well-known result
that $\SAT$ parameterized by the number of variables has no polynomial
kernel unless $\NP \subseteq \coNP/\text{\normalfont poly}$
\cite{BodlaenderDowneyFellowsHermelin09,FortnowSanthanam08}, and the
fact that $\var(F)$ is always a strong $\CCC$\hy backdoor set of $F$
if $\CCC$ is self-reducible.

Less immediate is the question whether  \textsc{$\CCC$\hy Backdoor Set
    Evaluation} admits a polynomial kernel if the inputs are restricted
  to 3CNF formulas, as 3SAT parameterized by the number of variables has
  a cubic kernel by trivial reasons.
However, for $\HORN$ and $\TWOCNF$ this question can be answered
negatively.

\begin{proposition}[\cite{Szeider11a}]
 \textsc{$\CCC$\hy Backdoor Set
    Evaluation} does not admit a polynomial kernel 
  for  $\CCC\in \{\HORN,\TWOCNF \}$
  unless $\NP \subseteq
  \coNP/\text{\normalfont poly}$, even if the input formula is in 3CNF.
\end{proposition}

\section{Backdoor Trees}

Backdoor trees are binary decision trees on
backdoor variables whose leaves correspond to instances of the base
class. Every strong backdoor set of size~$k$ gives rise to a backdoor
tree with at least~$k+1$ and at most~$2^k$ leaves.  It is reasonable
to rank the hardness of instances in terms of the number of leaves of
backdoor trees, thus gaining a more refined view than by just
comparing the size of backdoor sets.

Consider the CNF formula $F$ with variables~$x_1,\dots,x_{2n}$ and
$y_1,\dots,y_n$ consisting of all clauses of
the form
\[
\begin{array}{l}
\{y_i,\neg x_1,\dots,\neg x_{2i-2},
 x_{2i-1},
\neg x_{2i}, \dots, \neg x_{2n}\},\\
\{y_i,\neg x_1,\dots,\neg x_{2i-1},
 x_{2i}, \neg x_{2i+1}, \dots, \neg x_{2n}\},
\end{array}
\] 
for $1\leq i \leq n$.  The set $B=\{y_1,\dots,y_n\}$ is a strong
$\HORN$\hy backdoor set (in fact,~$B$ is the smallest possible).
However, every $\HORN$\hy backdoor tree~$T$
with~$\var(T)=\{y_1,\dots,y_n\}$ has $2^n$ leaves.  On the other hand,
the formula~$F$ has a $\HORN$\hy backdoor tree~$T'$ with only $2n+1$
leaves where~$\var(T')=\{x_1,\dots,x_{2n}\}$.  Thus, when we want to
minimize the number of leaves of backdoor trees, we must not restrict
ourselves to variables of a smallest strong backdoor set.

The problem \textsc{$\CCC$-Backdoor Tree Detection} now takes as input a
CNF formula~$F$, a parameter $k$, and asks whether $F$ has a
$\CCC$\hy backdoor tree with at most~$k$ leaves.
 
A base class $\CCC$ is said to admit a \emph{loss-free kernelization}
if there exists a polynomial-time algorithm that, given a CNF
formula~$F$ and an integer~$k$, either correctly decides that $F$ has no
strong $\CCC$\hy backdoor set of size at most~$k$, or computes a
set~$X\subseteq \var(F)$ such that the following conditions hold:
(i)~$X$ contains all minimal strong $\CCC$\hy backdoor sets of~$F$
  of size at most~$k$; and
(ii)~the size of~$X$ is bounded by a computable function that
depends on~$k$ only.

Samer and Szeider~\cite{SamerSzeider08b} have shown that \textsc{$\CCC$\hy
  Backdoor Tree Detection} is fixed-parameter tractable for every base
class $\CCC$ that admits a loss-free kernelization. Since Buss-type
kernelization is loss-free, the two classes $\HORN$ and $\TWOCNF$ admit
a loss-free kernelization.  Hence \textsc{$\CCC$\hy Backdoor Tree
  Detection} is fixed-parameter tractable for $\CCC\in
\{\TWOCNF,\HORN\}$.

\section{Backdoors for Problems Beyond NP}

The backdoor approach has been successfully applied to obtain
fixed-parameter tractability for problems whose unparameterized
worst-case complexity lies beyond~NP. In particular, FPT results have
been obtained for the $\#P$-complete problem Propositional Model
Counting, the PSPACE-complete QBF-SAT problem, and problems of
nonmonotonic reasoning and abstract argumentation that are located on
the second level of the Polynomial Hierarchy. In this section we briefly
survey these results.

\subsection{Propositional Model Counting}

The $\#\SAT$ problem asks to compute for a given CNF formula $F$ the
number of assignments $\tau \in 2^{\var(F)}$ that satisfy~$F$.  This
problem arises in several areas of Artificial Intelligence, in
particular in the context of probabilistic
reasoning~\cite{BacchusDalmaoPitassi03,Roth96}.  The problem is
$\#$P-complete and remains $\#$P-hard even for monotone 2CNF formulas and
Horn 2CNF formulas. It is $\NP$\hy hard to approximate the number of
satisfying assignments of a CNF formula with $n$ variables within
$2^{n^{1-\epsilon}}$ for any $\epsilon>0$.  This approximation hardness
holds also for monotone 2CNF formulas and Horn 2CNF
formulas~\cite{Roth96}.  However, if $\#\SAT$ can be solved in
polynomial time $O(n^c)$ for the formulas of a base class $\CCC$, and if we
know a strong $\CCC$\hy backdoor set of a formula $F$ of size $k$, then
we can compute the number of satisfying assignments of $F$ in time
$O(2^kn^c)$~\cite{NishimuraRagdeSzeider07,SamerSzeider08c}.  For some
applications in probabilistic reasoning one is interested in the
weighted model counting (WMC) problem, which is  more general than
$\#\SAT$ (see, e.g., \cite{SangBeameKautz05,ChaviraDarwiche08}). Since
the backdoor set approach applies also to the more general problem, we
will use it for the following discussions.

A \emph{weighting} $w$ of a CNF formula $F$ is a mapping $w$ that
assigns each variable $x\in \var(F)$ a rational number $0 \leq w(x) \leq
1$; this generalizes to literals by $w(\ol{x})=1-w(x)$ and to truth
assignments $\tau\in 2^X$ by $w(\tau)=\prod_{x\in X} w(x^{\tau(x)})$.
We define $\#_w(F)$ as the sum of the weights of all assignments $\tau\in
2^{\var(F)}$ that satisfy $F$.  The WMC problem asks  to compute
$\#_w(F)$ for a given CNF formula $F$ and weighting $w$. WMC is clearly
at least as hard as computing $\#(F)$ as we can reduce $\#\SAT$ to WMC
by using the weight 1/2 for all $n$ variables and multiplying the result
by $2^{n}$.  A strong $\CCC$\hy backdoor set $X$ of a CNF formula $F$ can be
used to compute $\#_w(F)$ via the equation
\[
\#_w(F)=\sum_{\tau \in 2^{X}} w(\tau) \cdot \#_w(F[\tau]). 
\] 
It is easy to see that WMC is polynomial for the base classes $\CLU$ and
$\FOREST$ as the corresponding algorithms for deciding satisfiability
for these classes as discussed above allow a straightforward
generalization to WMC.  From Theorem~\ref{the:strong-tree} and
Proposition~\ref{pro:clu-deletion} we conclude that WMC is fixed-parameter
tractable parameterized by $\sb_{\FOREST}$ and $\db_{\CLU}$.

\subsection{Quantified Boolean Formulas}

\newcommand{\QHORN}{\text{\normalfont\textsc{QHorn}}}
\newcommand{\QTWOCNF}{\text{\normalfont\textsc{Q2CNF}}}

Many important computational tasks like planning, verification, and
several questions of knowledge representation and automated reasoning
can be naturally encoded as the evaluation problem of \emph{quantified
  Boolean formulas
  (QBF)}~\cite{OtwellRemshagenTruemper04,Rintanen99,SabharwalAnsoteguiGomesHartSelman06}.
A QBF consists of a propositional CNF formula $F$ (the ``matrix'') and a
quantifier prefix. For instance $\mathcal{F} = \forall y \,\forall z
\,\exists x \,\exists w \,F$ with $F = \{\{\neg x, y, \neg w\}$,
$\{x,\neg y, w\}$, $\{\neg y, z\}$, $\{y , \neg z\}\}$ is a QBF.  The
evaluation of quantified Boolean formulas constitutes a {\sc
  PSPACE}-complete problem and is therefore believed to be
computationally harder than the NP-complete propositional satisfiability
problem~\cite{KleineBuningLettman99,Papadimitriou94,StockmeyerMeyer73}.
Only a few tractable classes of quantified Boolean formulas are known
where the number of \emph{quantifier alternations} is unbounded.  For
example, the time needed to solve QBF formulas whose primal graph has
bounded treewidth grows non-elementarily in the number of quantifier
alternations~\cite{PanVardi06}.  Two prominent tractable classes with
\emph{unbounded} quantifier alternations are $\QHORN$ and $\QTWOCNF$
which are QBFs where the matrix is a Horn or 2CNF formula, respectively.
$\QHORN$ formulas and $\QTWOCNF$ formulas can be evaluated in polynomial
time due to well-known results of Kleine B\"{u}ning \etal
\cite{KleinebuningKarpinskiFlogel95} and of Aspvall \etal
\cite{AspvallPlassTarjan79}, respectively.
 
In order to evaluate a QBF formula with a small strong $\HORN$- or
$\TWOCNF$-backdoor set $X$ efficiently, we require that $X$ is closed
under variable dependencies. That is, if $x$ depends on $y$ and $x\in
X$, then also $y\in X$, where we say that $x$ depends on $y$ if the
quantifier for $y$ appears to the left of the quantifier for $x$, and
one cannot move the quantifier for $y$ to the right of $x$ without
changing the validity of the QBF. In general deciding whether a variable
depends on the other is PSPACE complete, but there are
``over-approximations'' of dependencies that can be computed in
polynomial time. Such over-approximations can be formalized in terms of
\emph{dependency schemes}.  Indeed, it is fixed-parameter tractable to
detect strong $\HORN$ or $\TWOCNF$\hy backdoor sets of size at most $k$
that are closed \mywrt any fixed polynomial-time decidable
dependency scheme \cite{SamerSzeider09a}. 
  This fpt
  result allows an unbounded number of quantifier alternations for each
  value of the parameter, in contrast to the results for parameter
  treewidth.

\subsection{Nonmonotonic Reasoning}
\label{subsec:asp}

\newcommand{\aspnot}{\neg}

Answer-Set Programming (ASP) is an increasingly popular framework for
declarative programming~\cite{MarekTruszczynski99,Niemela99}. ASP allows
to describe a problem by means of rules and constraints that form a
disjunctive logic program $P$ over a finite universe $U$ of atoms. A
rule $r$ is of the form $(x_1\vee \dots \vee x_l \leftarrow
y_1,\dots,y_n,\aspnot z_1,\dots,\aspnot z_m)$.  We write
$\{x_1,\dots,x_l\}=H(r)$ (the \emph{head} of $r$) and
$\{y_1,\dots,y_n,z_1,\dots,z_m\}=B(r)$ (the \emph{body} of $r$),
$B^+(r)= \{y_1,\dots,y_n\}$ and $B^-(r)= \{z_1,\dots,z_n\}$.  A set $M$
of atoms \emph{satisfies} a rule $r$ if
$B^+(r)\subseteq M$ and $M\cap B^-(r) = \emptyset$ implies $M\cap H(r) \neq \emptyset$.
$M$ is a
\emph{model} of $P$ if it satisfies all rules of $P$.  The \emph{GL
  reduct} of a program $P$ under a set $M$ of atoms is the program $P^M$
obtained from $P$ by first removing all rules $r$ with $B^-(r)\cap M\neq
\emptyset$ and second removing all $\aspnot z$ where $z \in B^-(r)$ from
all remaining rules $r$~\cite{GelfondLifschitz91}. $M$ is an
\emph{answer set} of a program $P$ if $M$ it is a minimal model of
$P^M$.

For instance, from the program $P=\{(\text{tweety-flies} \leftarrow
\text{tweety-is-a-bird}, \linebreak[1] \aspnot \text{tweety-is-a-penguin}), \linebreak[1]
(\text{tweety-is-a-bird}\leftarrow )\}$ we may conclude that tweety-flies, since
this fact is contained in the only answer set $\{\text{tweety-is-a-bird},
\text{tweety-flies}\}$ of $P$.  If we add the fact tweety-is-a-penguin to the
program and obtain $P'=P\cup \{(\text{tweety-is-a-penguin}\leftarrow )\}$, then
we have to retract our conclusion tweety-flies since this fact is not
contained in any answer set of $P'$ (the only answer set of $P'$ is
$\{\text{tweety-is-a-bird},
\text{tweety-is-a-penguin}\}$). This \emph{nonmonotonic} behaviour that adding a fact may allow
fewer conclusions is typical for many applications in Artificial
Intelligence. The main computational
problems for ASP (such as deciding whether a program has a solution, or
if a certain atom is contained in at least one or in all answer sets)
are of high worst-case complexity and are located at the second level of
the Polynomial Hierarchy~\cite{EiterGottlob95}.

Also for ASP several islands of tractability are known, and it is
possible to develop a backdoor approach~\cite{FichteSzeider11}.  Similar
to SAT one can define partial truth assignments $\tau$ on a set of atoms and
solve a disjunctive logic program $P$ by solving all the reduced
programs $P[\tau]$. However, the situation is trickier than for
satisfiability.  Although every answer set of $P$ corresponds to an
answer set of $P[\tau]$ for some truth assignment $\tau$, the reverse
direction is not true.  Therefore, one needs to run a check for each
answer set of $P[\tau]$ whether it gives rise to an answer set of
$P$. Although this correctness check is polynomial, we must ensure
that we do not need to carry it out too often.  A sufficient condition
for bounding the number of checks is that we can compute all answer
sets of a program $P\in \CCC$ in polynomial time (``$\CCC$ is
enumerable''). In particular, this means that $P\in \CCC$ has only a
polynomial number of answer sets, and so we need to run the correctness
check only a polynomial number of times.

Several enumerable islands of tractability have been identified and
studied regarding the parameterized complexity of backdoor set
detection~\cite{FichteSzeider11}. For instance, programs where each rule
head contains exactly one atom and each rule body is negation-free are
well-known to have exactly one answer set. Such programs are called Horn
programs, and similar to satisfiability, one can use vertex covers to
compute backdoor sets \mywrt Horn.  Further enumerable islands
of tractability can be defined by forbidding cycles in graphs, digraphs,
and mixed graphs associated with disjunctive logic programs. Now, one
can use feedback vertex set (fvs) algorithms for the considered graphs
to compute backdoor sets: undirected fvs~\cite{DowneyFellows99},
directed fvs~\cite{ChenLiuLuOsullivanRazgon08}, and mixed fvs
\cite{BonsmaLokshtanov10}.  One can get even larger enumerable islands
of tractability by labeling some of the vertices or edges and by only
forbidding ``bad'' cycles, namely cycles that contain at least one
labeled edge or vertex. For the undirected case one can use subset
feedback vertex set algorithms to compute backdoor sets
\cite{CyganPilipczukPilipczukWojtaszczyk11,KawarabayashiKobayashi10}.
Currently it is open whether this problem is fixed-parameter tractable
for directed or mixed graphs.  Even larger islands can be obtained by
only forbidding bad cycles with an even number of labeled vertices or
edges \cite{Fichte11}. This gives rise to further challenging feedback
vertex set problems.

\subsection{Abstract Argumentation}

The study of arguments as abstract entities and their interaction in
form of \emph{attacks} as introduced by~Dung~\cite{Dung95} has become
one of the most active research branches within Artificial Intelligence,
Logic and
Reasoning~\cite{BenchcaponDunne07,BesnardHunter08,RahwanSimari09}.
Abstract argumentation provides suitable concepts and formalisms to
study, represent, and process various reasoning problems most
prominently in defeasible reasoning (see,
e.g.,~\cite{Pollock92,BondarenkoDungKowalskiToni97}) and agent
interaction (see, e.g., \cite{ParsonsWooldridgeAmgoud03}).  An abstract
argumentation system can be considered as a directed graph, where the
vertices are called ``arguments'' and a directed edge from $a$ to $b$
means that argument $a$ ``attacks'' argument $b$.

A main issue for any argumentation system is the selection of acceptable
sets of arguments, called extensions. Whether or not a set of arguments
is accepted is considered \mywrt certain properties of sets of
arguments, called semantics~\cite{BaroniGiacomin09}. For instance, the
\emph{preferred semantics} requires that an extension is a maximal set
of arguments with the properties that (i) the set is independent, and (ii)
each argument outside the set which attacks some argument in the set
is itself attacked by some argument in the set. Property~(i) ensures
that the set is conflict-free, property~(ii) ensures that the set
defends itself against attacks.

Important computational problems are to determine whether an argument
belongs to some extension (credulous acceptance) or whether it belongs
to all extensions (skeptical
acceptance)~\cite{DimopoulosTorres96,DunneBenchcapon02}.  For most
semantics, including the preferred semantics, the problems are located
on the second level of the Polynomial Hierarchy~\cite{Dunne07}.

It is known that the acceptance problems can be solved in polynomial
time if the directed graph of the argumentation framework is
\emph{acyclic}, \emph{noeven} (contains no even cycles),
\emph{symmetric}, or
\emph{bipartite}~\cite{Dung95,BaroniGiacomin09,CostemarquisDevredMarquis05,Dunne07}.
Thus, these four properties give rise to islands of tractability for
abstract argumentation, and one can ask whether a backdoor approach can
be developed to solve the acceptance problems for instances that are
close to an island. Here it is natural to consider deletion backdoor sets,
\ie, we delete arguments to obtain an instance that belongs to the
considered class. For the islands of acyclic, symmetric, and bipartite
argumentation frameworks we can find a backdoor using the
fixed-parameter algorithms for directed feedback vertex set
\cite{ChenLiuLuOsullivanRazgon08}, vertex cover \cite{DowneyFellows99}
and for graph bipartization~\cite{ReedSmithVetta04}, respectively.  For
finding a vertex set of size $k$ that kills all directed cycles of even
length we only know an XP algorithm which is based on a deep
result~\cite{RobertsonSeymourThomas99}.
 
However, it turns out that using the backdoor set is tricky and quite
different from satisfiability and answer set
programming~\cite{OrdyniakSzeider11}. The acceptance problems remain
\mbox{(co-)}NP-hard for instances that can be made symmetric or
bipartite by deleting one single argument.  On the other hand, if an
instance can be made acyclic or noeven by deleting $k$ arguments, then
the acceptance problems can be solved in time $3^kn^c$.  The base 3 of
the running time comes from the fact that 
the evaluation algorithm considers three
different cases for the arguments in the backdoor set:
(1)~the argument is in the acceptable set, (2)~the argument is
not in the set and is attacked by at least one argument from the set,
and (3)~the argument is
not in the set but is not attacked by any argument from the set.

\section{Conclusion}
 
Backdoor sets aim at exploiting hidden structures in real-world problem
instances.  The effectiveness of this approach has been investigated
empirically in
\cite{DilkinaGomesSabharwal07,FichteSzeider11,LiBeek11,SamerSzeider08b}
and in many cases, small backdoor sets were found for large industrial
instances.

As several backdoor set problems reduce to well-investigated core
problems from parameterized complexity, such as \textsc{Vertex Cover},
3-\textsc{Hitting Set}, \textsc{Feedback Vertex Set}, and their
variants, a few decades of focused research efforts can be used to
detect backdoor sets efficiently. Nevertheless, several questions remain
open. In particular, the parameterized complexity classification of
several permissive problems seems challenging. As discussed at the end
of Subsection \ref{subsec:asp}, the classification of variants of the
\textsc{Feedback Vertex Set} problem would also shed some light on
backdoor set detection problems in nonmonotonic reasoning.

We believe that more research in this direction is necessary if we want
to explain the good practical performance of heuristic SAT
solvers. Directions for future research could involve multivariate
parameterizations of backdoor problems and the consideration of
backdoors to combinations of different base classes.

\paragraph{Acknowledgment} We thank Ryan Williams for his comments on
an earlier version of this survey.
 
\bibliography{literature}
\bibliographystyle{plain}

\end{document}